\documentclass[acmnow]{acmtrans2m}

\usepackage{algorithm}
\usepackage{algorithmic}
\usepackage{amssymb}
\usepackage{graphicx}
\usepackage{subfigure}
\usepackage{color}
\usepackage{mathrsfs}
\usepackage{amsfonts}
\usepackage{amssymb}

\title{Click Efficiency: A Unified Optimal Ranking for Online Ads and Documents}

\author{
Raju Balakrishnan\\
       rajub@asu.edu \\
       School of Computing and Informatics\\
       Arizona State University, Tempe AZ 85287
        \and
       Subbarao Kambhampati \\
       rao@asu.edu \\
       School of Computing and Informatics\\
       Arizona State University, Tempe AZ 85287
       }

\begin{abstract}
Traditionally the probabilistic ranking principle is used to rank the search results while the ranking based on expected profits is used for paid placement of ads. These rankings try to maximize the expected utilities based on the  user click models. Recent empirical analysis on search engine logs suggests a unified click models for both ranked ads and  search results. The segregated view of document and ad rankings does not consider this commonality. Further, the used model considers parameters of (i) probability of the user abandoning browsing results (ii) perceived relevance of result snippets. But how to consider them for improved ranking is unknown currently. In this paper, we propose a generalized  ranking function---namely \emph{Click Efficiency (CE)}---for documents and ads based on empirically proven user click models. The ranking considers parameters (i) and (ii) above, optimal and has the same time complexity as sorting. To exploit its generality, we examine the reduced forms of CE ranking under different assumptions enumerating a hierarchy of ranking functions. Interestingly, some of the rankings in the hierarchy are currently used ad and document ranking functions; while others suggest new rankings. Thus this hierarchy illustrates the relations between different rankings, as well as clarifies the underlying assumptions.  While optimality of ranking is sufficient for document ranking, applying $CE$ ranking to ad auctions requires an appropriate pricing mechanism. We incorporate a second price based pricing mechanism with the proposed ranking. Our analysis proves several desirable properties including revenue dominance over VCG for the same bid vector and existence of a Nash equilibrium in pure strategies. The equilibrium is socially optimal, and revenue equivalent to the truthful  VCG equilibrium. As a result of its generality,  the auction mechanism and the  equilibrium reduces to the current mechanisms including GSP and corresponding equilibria. Further, we relax the independence assumption in CE ranking and analyze the diversity ranking problem. We show that optimal diversity ranking is NP-Hard in general, and that a constant time approximation algorithm is not likely.
\end{abstract}
\begin{document}

\maketitle
\section{Introduction}
Search engines rank the  results to maximize relevance of the top documents. On the other hand,  targeted ads are ranked to maximize the profit from clicks. In general, the users browse through ranked lists of results or ads from top to bottom either clicking or skipping the results, or abandoning browsing the list due to impatience or satiation. The  goal of the ranking is to maximize the expected  relevances (or profits) of clicked results based on the  click model of the users. The sort by relevance ranking suggested by Probability Ranking Principle (PRP) is commonly used for search results for decades~\cite{robertson:prp,michael:utility}. In contrast, sorting by the expected profits calculated as the product of bid amount and Click Through Rate (CTR) is popular for ranking ads~\cite{richardson:clickpred}.

  Recent click models suggest that the user click behaviors for both search results and targeted ads are the same~\cite{guo2009click,zhu2010novel}. Considering this commonality, the only difference between the two ranking problems is the utilities of entities ranked: for documents utility is the relevance and for the ads it is the cost-per-click. This suggests the possibility of a unified ranking function for results and ads. The current segregation of  document and ad ranking as separate areas does not consider this commonality. A unified approach can help to widen the   scope  of the related research to these two areas, and enable applications of  existing ranking functions in one area to isomorphic problems in the other area as we will show below.

In addition to the unified approach, the recent click models consider the following parameters:
 \begin{enumerate}
  \item \textbf{Browsing Abandonment: } The user may abandon browsing ranked list at any point. The likelihood of abandonment may depend on the entities the user has already seen~\cite{zhu2010novel}.  \label{paramAbandon}
 \item \textbf{Perceived Relevance: } Perceived relevance is the user's relevance assessment viewing only the search snippet or ad impression. The decision to click or not depends on the perceived relevance, not on the actual relevance of the results~\cite{yue-beyond,clarke:captionsfeatures}. \label{paramPerceived}
 \end{enumerate}
Though these parameters are part of the click models~\cite{guo2009click,zhu2010novel} how to exploit these parameters to improve ranking is unknown. The current document ranking is based on the simplifying assumption  that the perceived relevance is the same as the actual relevance of the document, and ignores the browsing abandonment. The ad placement partially considers perceived relevance, but ignores abandonment probabilities.

 We propose a unified optimal ranking function---namely \emph{Click Efficiency (CE)}---based on a generalized click model of the user. CE is defined as the ratio of the stand-alone utility generated by an entity to the sum of the abandonment probability and click probability of that entity (abandonment probability is the probability for the user to leave browsing the list after viewing the entity). The sum of the abandonment and click probability may be viewed as the click probability consumed by the entity. We derive the name Click Efficiency  based on this view---similar to the  definition of the mechanical efficiency of a machine as the ratio of  the output  to the input energy. We show that sorting in the descending order of CE of entities guarantees optimum ranking utility. We do not make assumptions on  the utilities of the entities, which may be assessed relevance  for documents or cost per click (CPC) charged based on the auction for ads. On plugging in the appropriate utilities---relevance for documents and CPC for the ads---the ranking specializes to document and ad ranking.

 As an implication of the generality, the proposed ranking will reduce to specific ranking problems on assumptions on the user behavior. We enumerate a hierarchy of ranking functions corresponding to the limiting assumptions on the click model. Most interestingly, some of these special cases correspond to the currently used  document and ad ranking functions---including PRP and sort by expected profit described above. Further, some of the reduced ranking functions suggest new rankings for special cases of the click model---like a click model in which the user never abandons the search, or the perceived relevance is approximated as the actual relevance. This hierarchy elucidates interconnection between different ranking functions and the assumptions behind the rankings. We believe that this will help in choosing the appropriate ranking function for a particular user click behavior.


\textbf{Ad Ranking: }An ad placement mechanism consists of a ranking and a pricing strategy. Hence to apply the CE ranking to ad placement, a pricing mechanism has to be associated. We incorporate a second price based pricing mechanism with the proposed ranking. Our analysis establishes many interesting properties of the proposed mechanism. Particularly, we state and prove the existence of a Nash Equilibrium in pure strategies. At this equilibrium the profits of the search engine and the total revenue of the advertisers is simultaneously optimized.  Like ranking, this is a generalized auction mechanism, and reduces to the existing GSP and Overture mechanisms under the same assumptions as that of the ranking. Further, the stated Nash Equilibrium is a general case of the equilibriums of these existing mechanisms. Comparing the mechanism properties with that of VCG~\cite{vickrey1961counterspeculation,clarke1971multipart,groves1973incentives}, we show that for the same bid vector search engine revenue for the CE mechanism  will be greater or equal to that of VCG. Further,  the revenue for the proposed equilibrium is equal to the revenue of truthful dominant strategy equilibrium of VCG.

\textbf{Diversity Ranking: }Our analysis so far was based on the assumption of parameter independence between the ranked entities. We relax this assumption and analyze the implications based on a specific well known problem---diversity ranking~\cite{carterette-analysis,agrawal2009diversifying,rafiei-diversifying}. Diversity ranking try to maximize the collective utility of top-$k$ ranked entities. For a ranked list, an entity will reduce residual utility of a similar entity below in the list. Though optimizing many of the specific ranking functions incorporating diversity is known to be NP-Hard~\cite{carterette-analysis}, an understanding of why this is an inherently hard problem is lacking. By analyzing a significantly general case, we show that even the very basic formulation  of diversity ranking is NP-Hard. Further we extend our proof showing that a constant ratio approximation algorithm is unlikely. As a benefit of the generality of ranking, these results are applicable both for ads and documents.

The contributions of the unified ranking, including both ad and document domains are:
\begin{enumerate}
\item Unified optimal ranking (CE ranking) based a generalized click model.
\item Optimal ranking considering abandonment probabilities for documents and ads.
\item Optimal Ranking considering perceived relevance of documents and ads.
\item A unified hierarchy of ranking functions and enumerating optimal rankings for different click models.
\item Analysis of general diversity ranking problem and hardness proofs.
\end{enumerate}
Contributions to ad placement are:
\begin{enumerate}
\item Design and analysis of a generalized ad auction mechanism incorporating pricing with CE ranking.
\item Proving existence of a  socially optimal Nash Equilibrium with optimal advertisers revenue as well as optimal search engine profit.
\item Proof of search engine revenue dominance over VCG for equivalent bid vectors, and equilibrium revenue equivalence to the truthful VCG equilibrium.
\end{enumerate}

The rest of this paper is organized as following. Next section reviews related work. Section~\ref{sec-clickModel} explains the click model used for our analysis. Subsequently we introduce our optimal ranking function, and discuss the intuitions and implications. In Section~\ref{sec-existing} reductions of ranking function to several document  and ad ranking functions under limiting assumptions are enumerated. Further we discuss several useful special cases of our ranking and assumptions under which they are optimal. In Section~\ref{sec-pricing} we incorporate a pricing strategy to design a complete auction mechanism for ads. Several  useful properties are established, including the existence of Nash equilibrium and revenue dominance over VCG. Section~\ref{sec-divRanking} explores the ranking considering mutual influences and proves our hardness results. Finally we conclude by discussing potential future research.

\section{Related Work}
\label{rel-work}

The impact of click models on ranking has been analyzed in ad-placement. Balakrishnan and Kambhampati~\cite{raju2008optimal} proposed the optimal ad ranking considering mutual influences. The ranking uses the same user model, but the paper considers only ads, and does not include generalization and ad auction mechanisms. Aggarwal~\emph{et al.}~\cite{aggarwal-sponsored}  as well as Kempe and  Mahdian~\cite{kempe2008cascade} analyze placement of ads using a Markovian click model.  The click model is similar except for that the abandoning is not modeled separately from continuing probability. These two papers  optimize the  sum of the revenues of the advertisers, instead of the optimizing search engine profits as we do for ads in this paper. Giotis and Karlin~\cite{giotis2008equilibria} extend this work by applying GSP pricing and analyzing the equilibrium. The GSP pricing and ranking lacks the optimality and generality properties we prove in this paper. Deng and Yu~\cite{deng2009new} extend this work by suggesting a ranking and pricing schema for the search engines and prove the existence of a Nash Equilibrium.  The ranking is a simpler bid based ranking (not based on CPC as in our case); and mechanism as well as equilibrium do not show optimality properties we prove in this paper. Kuminov and Tennenholtz~\cite{kuminov2009user} proposed a Pay Per Action (PPA) model similar to the click models and compared the equilibrium of GSP mechanism on the model with the VCG.  Ad auctions considering influence of other ads on conversion rates  are analyzed by Ghosh and Sayedi~\cite{ghosh2010expressive}.

The existing document ranking based on PRP~\cite{robertson:prp} claims that a retrieval order sorted on relevance leads to the largest number of relevant
documents in a set than any other policy. Gordon and Lenk \cite{michael:utility,michael:suboptimal} identified the required assumptions for the optimality of the ranking according to PRP. Our discussion on PRP may be considered as an independent formulation of assumptions under which PRP is optimal for web ranking.

User behavior studies in  click models  validate  the  ranking function introduced. There are a number of position based and cascade models studied recently~\cite{dupret2008user,craswel:cascade,guo2009click,chapelle2009dynamic,zhu2010novel}. In particular, General Click Model (GCM)~ by Zhu \emph{et al.}\cite{zhu2010novel} is interesting for us, since other click models are special cases of GCM.  Zhu \emph{et al.}~\cite{zhu2010novel} have listed assumptions under which the GCM would reduce to other click models. We will discuss the relations of our model to GCM below. Optimizing utilities of two dimensional placement  of search results has been studied by Chierichetti~\emph{et al.}~\cite{chierichetti2011optimizing}

Along with the current click models, there has been research on evaluating perceived relevance of the search snippets~\cite{yue-beyond} and ad impressions~\cite{clarke:captionsfeatures}. Research in this direction neatly complements our new ranking function by estimating the parameters required.

Diversity ranking has received considerable attention recently~\cite{agrawal2009diversifying,rafiei-diversifying}. The objective functions used to measure diversity by prior works  are known to be NP-Hard~\cite{carterette-analysis}. 

\section{Click Model}
\label{sec-clickModel}
\begin{figure}
\centering
\includegraphics[width=50mm, height=40mm,trim=57mm 195mm 108mm 45mm]{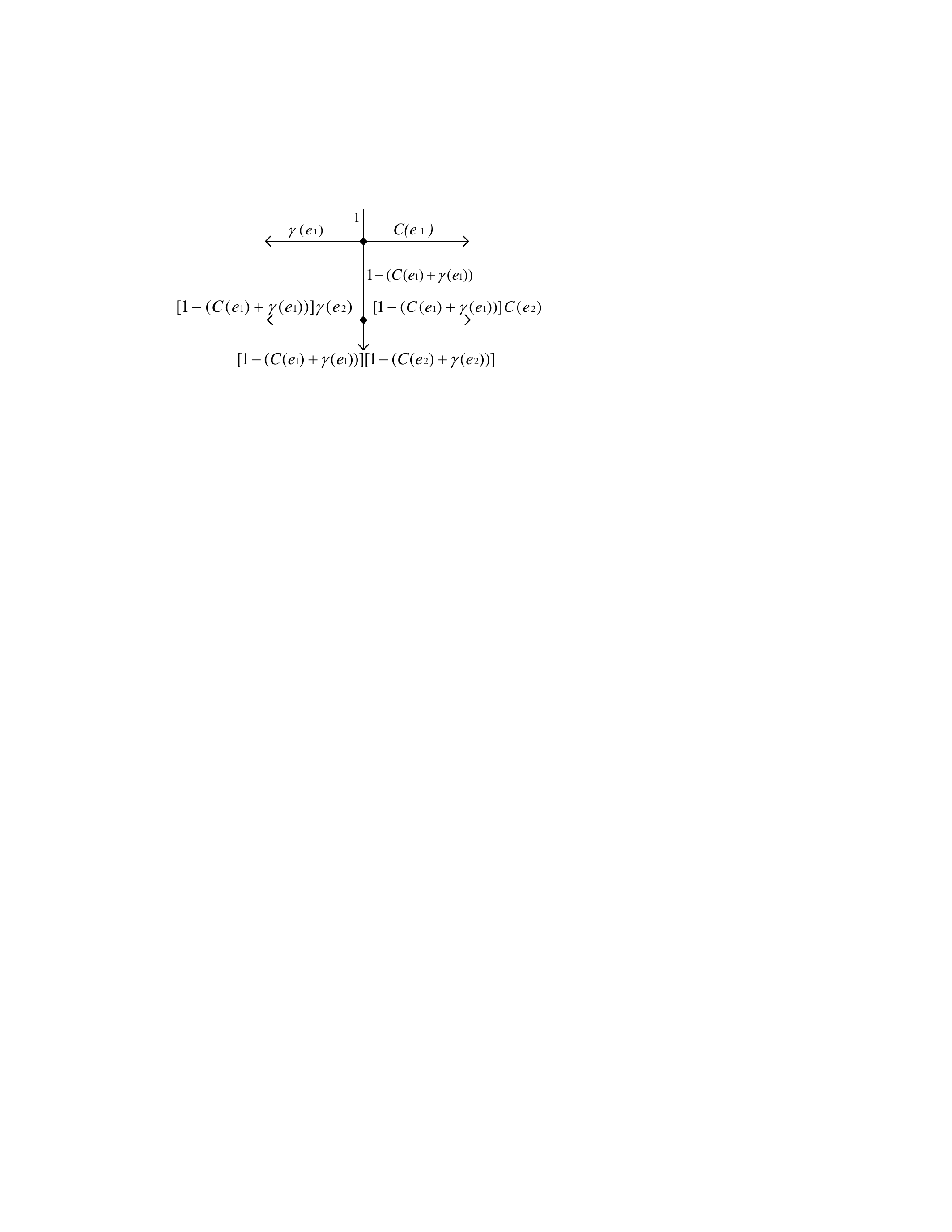}
\caption{Flow graph for an user browsing the first two entities. The labels
are the view probabilities and $e_i$ denotes the entity at the $i^{th}$
position}
\label{flowgraph-fig}
\end{figure}
\label{sec-usermodel}

We assume a basic user click model in which the web user browses the entity list in  ranked order, as shown in Figure~\ref{flowgraph-fig}. At every result entity, the user may:
\begin{enumerate}
\item Click the result entity with \emph{perceived relevance} $C(e)$. We define the perceived relevance as the probability to click the entity $e_i$ having seen $e_i$ i.e. $C(e_i)=P(click(e_i)| view(e_i))$. Note that Click Through Rate (CTR) defined in ad placement is the same as the perceived relevance defined here~\cite{richardson:clickpred}.
\item Abandon browsing the result list with \emph{abandonment probability} $\gamma(e_i)$.  $\gamma(e_i)$ is defined as the probability of abandoning the search at $e_i$ having seen $e_i$. i.e. $\gamma(e_i)=P(abandonment(e_i)|view(e_i))$.
\item Go to the next entity in the result list with probability $\big[1-\big(C(e_i)+\gamma(e_i) \big) \big]$\label{um-prob2}
\end{enumerate}

The click model can be schematically represented as the flow graph shown in Figure~\ref{flowgraph-fig}.
Labels on the edges refer to the probability of the user traversing
them. Each vertex in the figure corresponds to a view epoch (see
below), and the flow balance holds at each vertex. Starting from the
top entity, the probability of the user clicking the first ad is $R(e_1)$
and probability of him abandoning browsing is $\gamma(e_1)$. The user
goes beyond the first entity with probability $1-(R(e_1)+\gamma(e_1))$ and
so on for the subsequent results.

 In this model, we assume that the parameters---$C(e_i)$, $\gamma(e_i)$ and $U(e_i)$---are functions of the entity at the current position i.e. these parameters are independent of other entities the user has already seen. We recognize that this assumption is not fully accurate, since the users decision to click the current item or leave search may depend on not just on the current item but rather all the items he has seen before in the list. We stick to the assumption for the optimal ranking analysis below, since considering mutual influence of ads can lead to combinatorial optimization problems with intractable solutions. We will show that  even the simplest dependence between the parameters will indeed lead to intractable optimal ranking in Section~\ref{sec-divRanking}.

 Though the proposed model is intuitive enough, we would like to mention that the our  model is confirmed by the recent empirical click models. For example,  the General Click Model (GCM) by Zhu \emph{et al.}~\cite{zhu2010novel} is based on the same basic user behavior. The GCM is empirically validated for both search results and ads~\cite{zhu2010novel}. Further, other  click models are shown to be special cases of GCM (hence special cases of the model used in this paper).  Please refer to  Zhu \emph{et al.}~\cite{zhu2010novel} for a detailed discussion. These previous works avoids the need for separate model validation, as well as confirms feasibility of the parameter estimation.

\section{Optimal Ranking}
\label{sec-optRank}
Based on the click model, we formally define the ranking problem
and derive optimal ranking in this section.  The formal problem
statement is,

{\em Choose the optimal ranking $E_{opt}=\langle
e_{1},e_{2},..,e_{N}\rangle$ of $N$ entities to maximize the expected utility
\begin{equation}E(U)=\sum_{i=1}^{N}U(e_i)P_c(e_i)\label{objective-eq}\end{equation} where $N$ is
the total number of entities to be ranked.}

For the  browsing model in Figure~\ref{flowgraph-fig},  the click probability for the entity at the $i^{th}$ position is,
\begin{equation}P_c(e_i) = C(e_i) \prod_{j=1}^{i-1} \left[1 - \big( C(e_j) + \gamma(e_j)\big) \right] \label{eqn-umodel}\end{equation}

Substituting click probability $P_c$ from Equation \ref{eqn-umodel} in
Equation~\ref{objective-eq} we get,
\begin{equation}
E(U) = \sum_{i=1}^{N}U(e_i) C(e_i) \prod_{j=1}^{i-1} \left[ 1 - ( C(e_j) + \gamma(e_j)) \right] \label{selection-prob-eq}
\end{equation}

The optimal ranking maximizing this expected utility can be shown to be
a sorting problem with a simple ranking function:
\newtheorem{thm}{Theorem} \begin{thm}The expected utility in Equation~\ref{selection-prob-eq} is  maximum if the
entities are placed in the descending order of the value of the ranking function $CE$,
\begin{equation} CE(e_i)=\frac{U(e_i)C(e_i)}{C(e_i)+\gamma(e_i)} \label{eqn-RF}\end{equation}
\label{thmRF}
\end{thm}
\emph{Proof Sketch:}
The proof shows that any inversion in this order will reduce the expected profit. $CE$ function is  deduced from expected profits of two placements---the $CE$ ranked placement and placement in which the order of two adjacent ads are inverted. We show that the expected profit from the inverted placement can be no greater that the $CE$ ranked placement. Please refer to Appendix~\ref{appendixRF} for the complete proof.
$\Box$

As mentioned in the introduction, the ranking function  $CE$ is the utility generated per  unit view
probability consumed by the entity.  With respect to browsing model in Figure~\ref{flowgraph-fig},  the top entities in
the ranked list have higher view probabilities, and placing ads with greater utility per consumed view
probability higher intuitively increases total utilities.

Note that the ordering above does not maximize the utility for selecting a subset of items.  The  seemingly intuitive method of ranking the  set of items by $CE$ and selecting top-$k$ may not be optimal~\cite{aggarwal-sponsored}. For optimal selection, the  proposed ranking can be extended by a dynamic programming  based selection---similar to the method suggested by Aggrawal \emph{et al}~\cite{aggarwal-sponsored} for maximizing advertiser's profit. In this paper,  we discuss only the ranking problem.

\section{Ranking Taxonomy}
\label{sec-existing}
\begin{figure*}[t]
\centering
\includegraphics[width=130mm, height=75mm,trim=0mm 145mm 0mm 5mm]{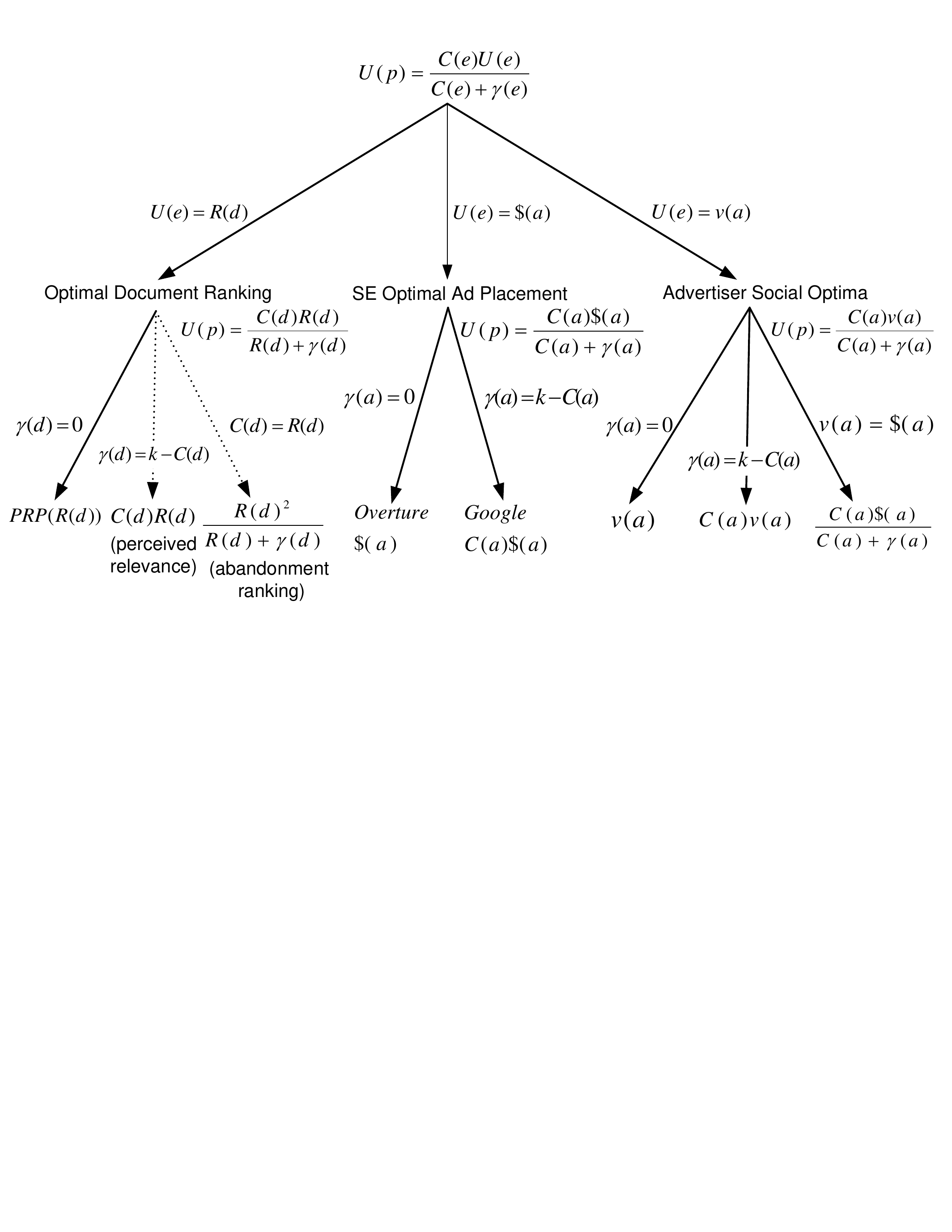}
\caption{Taxonomy reduced ranking functions of CE . The assumptions and corresponding reduced ranking functions are illustrated. The dotted lines denote predicted ranking functions incorporating new click model parameters. }
\label{fig-taxonomy}
\end{figure*}
As we mentioned before, the CE ranking will be applicable to different ranking problems by plugging in corresponding utilities. For example, if we plug in relevance as utility ($U(e)$ in Equation~\ref{eqn-RF}), the ranking function is for the documents, whereas if we plug in cost per click of ads, the ranking function is for ads. Further, we may assume specific constraints on one or more of the three parameters of CE ranking (e.g. $\forall_i\gamma(e_i)=0$). On these assumptions, CE ranking will suggest number of reduced ranking functions with specific applications. These substitutions and reductions can be enumerated as a taxonomy of ranking functions.

We show the taxonomy in Figure~\ref{fig-taxonomy}. The three top branches of the taxonomy ($U(e)=R(d)$, $U(e)=\$(a)$, and $U(e)=v(a)$ branches) are document, ad ranking maximizing search engine profit, and ad ranking maximizing advertisers revenue respectively. These  branches correspond to the substitution of utilities by document relevance, CPC, and private value of the advertisers. The sub-trees below these branches are the further reduced cases of these three main categories. The solid lines in  Figure~\ref{fig-taxonomy} denote already known functions, while the dotted lines are the new ranking functions suggested by  CE ranking. Sections~\ref{subsec-docRanking}, Section~\ref{subsec-adRanking}, and Section~\ref{subsec-socialOptimaRanking} below  discuss the further reductions of document ranking, search engine optimal ad ranking, and social optimal ad ranking respectively.

\subsection{Document Ranking}
\label{subsec-docRanking}
For the document ranking the utility of ranking is the probability of relevance of the document. Hence by substituting the document relevance---denoted by $R(d)$---in  Equation~\ref{eqn-RF} we get
\begin{equation} CE(d)=\frac{C(d)R(d)}{C(d)+\gamma(d)}\label{eq-RFDoc}\end{equation}
This function suggests the general optimal relevance ranking for the documents. We discuss some intuitively valid assumptions on user model for the document ranking and the corresponding ranking functions below. The three assumptions discussed below correspond to the three branches under Document Ranking subtree in Figure~\ref{fig-taxonomy}.

\medskip
\noindent{\bf  Sort by Relevance (PRP): }We elucidate two sets of assumptions under which the $CE(d)$ in Equation~\ref{eq-RFDoc} will reduce to PRP.

First  assume that the user has infinite patience, and never abandons results (i.e. $\gamma(d) \approx 0$).
Substituting this  assumption in Equation~\ref{eq-RFDoc},
\begin{equation} CE(d)\approx\frac{R(d)C(d)}{C(d)}=R(d)\label{eq-PRP}\end{equation}
which is exactly the ranking suggested by PRP.

In other words, scenarios in which the  user has infinite patience and never abandons checking the results (i.e. the user leaves browsing the results only by clicking a result) the PRP  is  still optimal.

Second set of slightly weaker assumptions under which the $CE(d)$ will reduce to PRP is
\begin{enumerate}
\item $C(d) \approx R(d)$.
\item Abandonment probability $\gamma(d)$ is negatively proportional to the document relevance i.e. $\gamma(d) \approx k-R(d)$, where $k$ is a constant  between one and zero. This assumption corresponds to the intuition that the higher the perceived relevance of the current result, the less likely is the user to abandon the search.
 \end{enumerate}
 Now the $CE(d)$  reduces to,
\begin{equation} CE(d) \approx \frac{R(d)^2}{k}\label{eq-PRP2}\end{equation}
Since this function is strictly increasing with $R(d)$, ordering by $R(d)$ results in the same ranking as suggested by the function. This  implies that PRP is optimal under these assumptions also.

 Abandonment  decreasing with perceived relevance is a more intuitively valid assumption than the infinite patience assumption above.

\medskip
\noindent{\bf  Ranking Considering Perceived Relevance: }Recently click log studies effectively assesses perceived relevance of document search snippets~\cite{yue-beyond,clarke:captionsfeatures}. But how to use the perceived relevance for improved document ranking is unknown. We show that  depending on the nature of abandonment probability $\gamma(d)$,  the optimal ranking considering perceived relevance differs.

If we assume that $\gamma(d) \approx 0$ in Equation~\ref{eq-RFDoc}, the optimal perceived relevance ranking  is the same as that suggested by  PRP as we have seen in Equation~\ref{eq-PRP}.

On the other hand, if we assume that the abandonment probability is negatively proportional to the perceived relevance ($\gamma(d)=k-C(d)$) as above, the optimal ranking considering perceived relevance is
\begin{equation} CE(d) \approx \frac{C(d)R(d)}{k} \propto C(d)R(d)\label{eq-RFPerc2}\end{equation}
i.e. sorting in the order of product of document relevance and perceived relevance is optimal under these assumptions.
The assumption of abandonment probabilities negatively proportional to relevance is  more realistic than infinite patience assumption as we discussed above. This discussion shows that by estimating nature of abandonment probability, one would be able to decide on  optimal perceived relevance ranking.

\medskip
\noindent{\bf  Ranking Considering Abandonment: } We  consider the ranking considering abandonment probability $\gamma(d)$, with the assumption that the perceived relevance is approximately equal to the actual relevance.
In this case the $CE(d)$ becomes,
\begin{equation} CE(d)\approx \frac{R(d)^2}{R(d)+\gamma(d)}\label{eq-RFAband}\end{equation}
Clearly this is not a strictly increasing function with $R(d)$. This means that ranking considering abandonment is different from PRP ranking, even if we assume that the perceived relevance is equal to the actual relevance. On the assumption  that $\forall_d\gamma(d)=0$, the abandonment ranking becomes same as PRP.

\subsection{Optimal Ad Ranking for Search Engines}
\label{subsec-adRanking}
For the paid placement of ads, the utility of ads to the search engine are  Cost Per Click (CPC) of ads.  Hence, by substituting the CPC of the ad---denoted by $P(a)$--- in  Equation~\ref{eqn-RF} we get
\begin{equation} CE(a)=\frac{C(a)P(a)}{C(a)+\gamma(a)}\label{eq-RFAd}\end{equation}
Thus this function suggests the general optimal ranking for the ads. Please recall that the perceived relevance $C(a)$ is the same as the Click Through Rate (CTR) used for ad placement~\cite{richardson:clickpred}.

In the following subsections demonstrate how the general ranking presented reduces to the currently used ad placement  strategies under appropriate assumptions. We will
show that they  all correspond to the  specific assumptions on the abandonment probability $\gamma(a)$. These two functions below correspond to the two branches under the SE Optimal ad Placement subtree in Figure~\ref{fig-taxonomy}.

\medskip
\noindent{\bf Ranking by Bid Amount:}  The sort by bid amount ranking was
 used by  Overture Services (and was later used by Yahoo! for a while after their acquisition of Overture).
Assuming  that the user never abandons browsing (i.e.
$\forall_a \gamma( a ) = 0 )$, Equation~\ref{eq-RFAd} reduces to
\begin{equation}CE( a) = P(a)\label{eq-adOverture}\end{equation}
This means that the ads are ranked purely in terms of their payment\footnote{Overture ranking is by bid amount, which is different from payment in a second price auction (since payment will be next higher bid amount). But both will result in the same ranking.}.

When $\gamma ( a ) = 0$, we essentially have a user with infinite
patience who will keep browsing downwards until he finds the relevant
ad. So, to maximize profit, it makes perfect sense to rank ads by bid
amount. More generally, for small abandonment probabilities, ranking
by bid amount  is near optimal. Note that this ranking is isomorphic to PRP ranking discussed above for document ranking, since both are ranking based only on utilities.

\medskip
\noindent{\bf Ranking by Expected Profit:} Google and Microsoft are purported to be placing the ads in the order of expected profit based on product of CTR ($C(a)$ in $CE$) and bid amount ($P(a))$~\cite{richardson:staticpr}. The ranking is part of the well known Generalized Second Price (GSP) auction mechanism. If we approximate abandonment probability as negatively proportional to the CTR of the ad (i.e. $\forall_a\gamma(a)=k-C(a)$) ,  the Equation~\ref{eq-RFAd} reduces to,
 \begin{equation}CE( a ) \approx \frac{ P(a)R(a)}{k} \propto P(a)R(a) \label{eq-RFAdGoogle}\end{equation}
This shows that ranking ads by  their stand-alone expected profit is near optimal as long as the abandonment probability is negatively proportional to the
relevance.\footnote{Google mechanism---GSP---uses bid amount of the advertisers (instead of CPC in Equation~\ref{eq-RFAdGoogle}) for ranking. We will show that both will result in the same ranking by an order preserving property of the GSP pricing in Section~\ref{sec-pricing}.} Note that this ranking is isomorphic to the perceived relevance ranking of the documents discussed above.

\subsection{Revenue Optimal Ad Ranking}
\label{subsec-socialOptimaRanking}
An important property of the auction mechanism is the expected revenue---which is the sum of the profits of the advertisers and search engine. To analyze advertisers profit, a private value model is commonly used. Each advertiser is assigned with a private value for the click equal to the expected revenue from the click. Advertisers pay a fraction of this revenue to the search engine depending on the pricing mechanism.  The profit for  advertisers is the difference  between the private value and payment to the search engine. The profit for the search engine is the payment of the advertisers. Consequently, the revenue is the sum of the profits of all the parties---search engine and the advertisers.

The Advertiser Social Optima branch in Figure~\ref{fig-taxonomy} corresponds to the ranking to maximize  total revenue. Private value  of advertisers $a_i$ is denoted as---$v(a_i)$. By substituting the utility by private values in Equation~\ref{eqn-RF} we get,
\begin{equation} CE(d)=\frac{C(a)v(a)}{C(a)+\gamma(a)}\label{eq-RFSocial}\end{equation}
If the ads are ranked in this order, the ranking will guarantee maximum revenue. This result is already known, as Agarwal~\emph{et al.}~\cite{aggarwal-sponsored} and Kempe and Mahdian~\cite{kempe2008cascade} independently analyzed the auctions to maximize advertiser revenue and proposed similar ranking function.

In Figure~\ref{fig-taxonomy} the two left branches of revenue maximizing subtree (labeled $\gamma(a)=0$ and $\gamma(a)=k-C(a)$) correspond to the assumption of no abandonment, and abandonment probabilities negatively proportional to the click probability respectively. These two cases are isomorphic to the Overture and Google ranking discussed  in Section~\ref{subsec-adRanking} above. We discuss further on revenue maximizing ranking in conjunction with a pricing mechanism in Section~\ref{sec-pricing}

\section{Applying $CE$ Ranking for Ad placement}
\label{sec-pricing}
We have shown that $CE$ ranking maximizes the profits for search engines for given CPCs in Section~\ref{subsec-adRanking}. In ad placement, the net profits  of ranking to the search engine can only be analyzed in association with a pricing mechanism. To this end, we introduce a pricing to be used with the $CE$ designing a full auction mechanism. Subsequently, we analyze the properties of the  mechanism.

 To describe the dynamics of ad auctions briefly, the search engine decides the ranking and pricing (cost per click) of  the ads based on the bid amounts of the advertisers. Generally the pricing  is not equal to the bid amount of advertisers, but is instead derived based on the bids~\cite{easley2010networks,edelman2005internet,aggarwal2006truthful}. In response to these ranking and pricing strategies, the advertisers (more commonly, the software agents of the advertisers) may change their bids to maximize their profits. They may change bids hundreds of times a day. Eventually, the bids will stabilize at a fixed point where no advertiser  can increase his profit by unilaterally changing his bid. This set of bids corresponds to a Nash Equilibrium of the auction mechanism. Hence the realistic profits of a search engine will be the profits corresponding to the Nash Equilibrium.

 The next section discusses  the properties of any mechanism based on the user model in Figure~\ref{flowgraph-fig}---independent of the ranking and pricing strategies. In Section~\ref{subsec-pricingStrat}, we introduce a pricing mechanism and analyze its properties including the equilibrium.
\subsection{Pricing Independent Properties}
\label{subsec-pricingIndAnalysis}
In this section we illustrate properties arising based on the user browsing model in Figure~\ref{flowgraph-fig}, not assuming any pricing or ranking strategy. One of the  basic results is
\newtheorem{remark}{Remark}
\begin{remark}
In any equilibrium the payment by the advertisers is less  than or equal to their private values (i.e. individual rationality of the bidders is maintained). \label{rem-IndRationality}
\end{remark}
If this is not true, advertiser may opt out from the auction by bidding zero and increase the profit, violating the assumption of an equilibrium.

\begin{remark}
In any equilibrium, price paid  by an advertiser increases monotonically as he  moves up in the ranking unilaterally.
\end{remark}
From the browsing model, click probability of the advertisers in non-decreasing as he moves up in the position. Unless the price increases monotonically, advertiser can increase  his profit by moving up, violating the assumption of an equilibrium.

Note that the proposed model is a general case of the positional auctions model by Varian~\cite{varian2007position}. Positional auctions assume static click probabilities for each position independent of the other ads. We assume realistic dynamic click probabilities depending on the ads above. Due to these externalities, the model is more complex and does  not hold many of the  properties derived by Varian~\cite{varian2007position} (e.g. monotonically increasing values and prices with positions).

\begin{remark}Irrespective of the ranking and pricing, the sum of revenues of the advertisers is upper bounded at
\begin{equation}
E(V) = \sum_{i=1}^{N}v(a_i) C(a_i) \prod_{j=1}^{i-1} \left[ 1 - ( C(a_j) + \gamma(a_j)) \right] \label{eqn-maxRevenue}
\end{equation}
when the advertisers are ordered by $\frac{C(a)v(a)}{C(a)+\gamma(a)}$. Further, this is an upper bound for the search engine profit.
\label{rem-revenueUpperBound}
\end{remark}
This result directly follows from the Advertisers Social Optima branch in Figure~\ref{fig-taxonomy}, and Equation~\ref{eq-RFSocial}.

The revenue is shared among the advertisers and search engine. For each click, advertisers get a revenue equal to the private value $v(a)$ and pay a fraction equal to the CPC (set by the search engine pricing strategy) to the search engine. The total payoff for the search engine is the sum of the payments by the advertisers. Conversely,  total payoff to the advertisers is the difference between total revenue and payoff to the search engine. Since the suggested order above in Remark~\ref{rem-revenueUpperBound} maximizes revenue, which is the sum of the payoffs of all the players (search engine and the advertisers), this is a socially optimal order and the revenue realized is the socially optimal revenue.

A corollary of the social optimality combined with the individual rationality result in Remark~\ref{rem-IndRationality} is that,
\begin{remark}
The quantity $E(V)$ in Remark~\ref{rem-revenueUpperBound} is an upper bound for the search engine profit irrespective of the ranking and pricing mechanism.
\label{rem-profitUpperBound}
\end{remark}

Social optimal revenue can be realized only if the ads are in the descending order of $\frac{C(a)v(a)}{C(a)+\gamma(a)}$. Social optimum  is desirable for search engines, since this will increase the payoffs for advertisers for the same CPC. Increased payoffs will increase the advertiser's incentive to advertise with the search engine and will increase business for the search engine in long term.

Since search engines do not know the private value of the advertisers (note that search engine perform the ranking), social optimal ranking based on private values is not directly feasible. We need to design a mechanism having an equilibrium coinciding with the social optimality. This will motivate advertisers towards bids coinciding with social optimal ordering. In addition to social optimality, it is highly desirable for the mechanism to be based on CE ranking to simultaneously maximize advertiser's revenue and  search engine profit. In the following section we propose such a mechanism using CE ranking and prove the existence of an equilibrium in which the CE ranking coincides with the social optimal allocation.

\subsection{Pricing and Equilibrium}
\label{subsec-pricingStrat}
In this section, we define a pricing strategy to  use with the CE ranking, and analyze properties of the resulting mechanism.

For defining the pricing strategy, we define the pricing order as the decreasing order of  $w(a)b(a)$, where $w(a)$ is,
\begin{equation}w(a) = \frac{C(a)}{C(a)+\gamma(a)} \label{eq-pricingOrder}\end{equation}
In this pricing order, we denote  the $i^{th}$ advertiser's $w(a_i)$ as $w_i$, $C(a_i)$  as $c_i$, $b(a_i)$  as $b_i$, and abandonment probability $\gamma(a_i)$ as $\gamma_i$ for convenience. Let $\mu_i= c_i+\gamma_i$. For each click advertiser $a_i$ is charged with a price $p_i$ (CPC) equal to the minimum bid required to maintain its position in the pricing order,
\begin{equation}p_i =\frac{w_{i+1}b_{i+1}}{w_i} = \frac{b_{i+1}c_{i+1}{\mu_i}}{\mu_{i+1}c_i} \label{eq-nextPricePricing} \end{equation}

Substituting $p_i$ in Equation~\ref{eq-RFAd} for the ranking order,  CE of the $i^{th}$ advertiser is,
\begin{equation}CE_i = \frac{p_ic_i}{\mu_i}\end{equation}

This proposed mechanism preserves the pricing order in the ranking order as well, i.e.
\begin{thm} The order by $w_ib_i$ is the same as the order by $CE_i$ for the auction i.e.
\begin{equation} w_i b_i \ge w_jb_j  \Longleftrightarrow CE_i \ge CE_j  \end{equation}
\label{thm-CEsameWi}
\end{thm}
Proof is given in the Appendix~\ref{appendix-CESameWi}. This order preservation property implies that the final ranking is the same as that based on bid amounts. As a corollary, the CPC is equal to the minimum amount the advertisers have to pay to maintain his position in the ranking order as well.

Further we show below that any advertisers CPC is less or equal to his bid.
\newtheorem{lemma}{Lemma}\begin{lemma}[(Individual Rationality)] The payment $p_i$ of any advertiser is less or equal to his bid amount.\label{thm-bidMaxPay}\end{lemma}
\begin{proof}
\begin{displaymath} p_i = \frac{b_{i+1}c_{i+1}\mu_i}{\mu_{i+1}c_i} =  \frac{b_{i+1}c_{i+1}}{\mu_{i+1}}\frac{\mu_i}{c_ib_i} b_i \le b_i \mbox{(since $CE_i \ge CE_{i+1}$)} \end{displaymath}
\end{proof}
This means advertisers will never have to pay more than bid, similar to GSP. This nice property makes it easy for the advertiser to decide his bid.

Interestingly, this mechanism also is a general case of the existing mechanisms, as in the case of CE ranking. In particular, the mechanism reduces to GSP (Google mechanism) and Overture mechanisms on the same assumptions on which CE ranking reduces to respective rankings (described in Section~\ref{subsec-adRanking}).

\begin{lemma}
The mechanism reduces to Overture ranking with second price auction on the assumption $\forall_{i} \gamma_i=0$\label{lemma-overRanking}
\end{lemma}
\begin{proof}
This assumption implies
\begin{eqnarray} w_i &= & 1  \nonumber \\
 & \Rightarrow & p_i = b_{i+1}\;\mbox{(second price auction)}\nonumber \\
 & \Rightarrow &  CE_i =  b_{i+1} \equiv b_i \;\mbox{(as  ranking is the the same)} \nonumber
\end{eqnarray}
\end{proof}

\begin{lemma}The mechanism reduces to GSP on assumption $\forall_{i} \gamma_i=k-c_i$\label{lemma-overRanking}\end{lemma}
\begin{proof}
This assumption implies
\begin{eqnarray} w_i &= & c_i  \nonumber \\
 & \Rightarrow & p_i = \frac{b_{i+1}c_{i+1}}{c_i}\;\mbox{(GSP pricing)} \nonumber \\
 & \Rightarrow &  CE_i =  \frac{b_{i+1}c_{i+1}}{k} \equiv \frac{b_{i}c_{i}}{k} \;\mbox{ (by Theorem~\ref{thm-CEsameWi})} \nonumber \\
 & \propto &  b_ic_i\nonumber
\end{eqnarray}
\end{proof}

%
%

This in conjunction with Theorem~\ref{thm-CEsameWi} implies  that GSP ranking by $c_ib_i$ (i.e. by bids) is the same as the ranking by $c_ip_i$ (by CPCs).

Now we will look at the equilibrium properties of the mechanism. Truth telling is not a dominant strategy. This trivially follows, since GSP is a special case of the proposed mechanism, and it is known that for GSP truth telling is not a dominant strategy~\cite{edelman2005internet}. Hence we center our analysis on Nash Equilibrium conditions. 

\begin{thm}[(Nash Equilibrium)]
Without the loss of generality assume that the advertisers are ordered in the decreasing order of $\frac{c_iv_i}{\mu_i}$ where $v_i$ is the private value of the $i^{th}$ advertiser. The advertisers are in an envy free pure strategy Nash Equilibrium if
\begin{equation}b_i = \frac{\mu_i}{c_i}\left[ v_i c_i + (1-\mu_i) \frac{b_{i+1}c_{i+1}}{\mu_{i+1}} \right] \label{eq-nashEqBids} \end{equation}
This equilibrium is socially optimal as well as optimal for search engines for the given CPC's.
\label{thmNash}
\end{thm}
\emph{Proof Sketch:}
The inductive proof shows that for these bid values, no advertiser can increase his profit by moving up or down in the ranking. The full proof is  given in Appendix~\ref{appedixNash}.
$\Box$

 We do not rule out the existence of multiple equilibria. The stated equilibrium is particularly  interesting, due to the simultaneous social optimality and search engine optimality.

The following remarks show that equilibria of other placement mechanisms are  reduced cases of the proposed CE equilibrium, as a natural consequence of its generality. The stated equilibrium reduces to equilibriums in Overture mechanism and GSP under the same assumptions under which the ranking reduces to respective rankings.

\begin{remark}
The bid values
\begin{equation}b_i =  v_i c_i + (1-c_i) b_{i+1}  \end{equation}
 are a pure strategy Nash Equilibrium in Overture mechanism. This corresponds to the  substitution of the assumption  $\forall_{i} \gamma_i=0\;(i.e.\ \mu_i = c_i)$ in Theorem~\ref{thmNash}.
\label{rem-overtEquSpecialCase}
\end{remark}
The proof  follows from  Theorem~\ref{thmNash} as both pricing and ranking is shown to be a special case of our proposed mechanism.

Similarly for GSP,
\begin{remark}
The bid values
\begin{equation}b_i =  v_i k  + (1-k) b_{i+1}c_{i+1}  \end{equation}
is a pure strategy Nash Equilibrium in GSP mechanism.
\label{rem-overtEquSpecialCase}
\end{remark}
 This equilibrium corresponds to the  substitution of the assumption  $\forall_{i} \gamma_i=k-c_i\;(1 \ge k \ge 0)$ in Theorem~\ref{thmNash}. Since this is a special case, the proof for Theorem~\ref{thmNash} is sufficient.


\subsection{Comparison with VCG mechanism}
We compare the revenue and equilibrium of  $CE$ mechanism with those of VCG~\cite{vickrey1961counterspeculation,clarke1971multipart,groves1973incentives}.
VCG auctions combine an optimal allocation (ranking) with VCG pricing.  VCG payment of a bidder is equal to the reduction of revenues of other bidders due to the presence of the bidder. A well known property is that VCG pricing with any socially optimal allocation has truth telling as the the dominant strategy equilibrium.

In the context of online ads, ranking optimal with respect to the bid amounts is socially optimal ranking for VCG. This optimal ranking is $\frac{b_ic_i}{\mu_i}$; as directly implied by the Equation~\ref{thmRF} on substituting $b_i$ for utilities. Hence this ranking combined with VCG pricing has truth telling as the dominant strategy equilibrium. Since  $b_i=v_i$ at the dominant strategy equilibrium,  ranking is socially optimal for an advertiser's true value as suggested in Equation~\ref{eq-RFSocial}.

The CE ranking function is different from VCG since  CE ranking by payments optimizes search engine profits. On the other hand,  VCG ranks by bids optimizing advertiser's profit. But the Theorem~\ref{thm-CEsameWi} shows that for the pricing used in $CE$, ordering of $CE$ is the same as that of VCG. This order preserving property facilitates  comparison of $CE$ with VCG. The theorem below shows revenue dominance of CE over VCG for the same bid values of advertisers.

\begin{thm}[(Search Engine Revenue Dominance)]
For the same bid values for all the advertisers, the revenue of search engine by $CE$ mechanism is greater or equal to the revenue by VCG.
\label{thm-CERevnueDominanceVCG}
\end{thm}
\emph{Proof Sketch:}
The proof is an induction  based on the fact that the ranking by CE and VCG are the same, as mentioned above. Full proof is given in Appendix~\ref{appendix-CERevnueDominanceVCG}.
$\Box$

This theorem shows that the CE mechanism is likely to provide higher revenue to the search engine even during transient times before the bids  settle on equilibriums.

Based on Theorem~\ref{thm-CERevnueDominanceVCG} we prove revenue equivalence of the proposed $CE$ equilibrium with dominant strategy equilibrium of VCG.

\begin{thm}[(Equilibrium Revenue Equivalence)]
At the equilibrium in  Theorem~\ref{thmNash}, the revenue of search engine  is equal to the revenue of the truthful dominant strategy equilibrium of VCG.
\label{thm-VCGrevnueEqiv}
\end{thm}

\emph{Proof Sketch:}
The proof is an inductive extension of the  of Theorem~\ref{thm-CERevnueDominanceVCG}. Please see Appendix~\ref{appendix-VCGRevnueEquiv} for complete proof.
$\Box$

Note that the $CE$ equilibrium has lower bid values than VCG at the equilibrium, but provides the same profit to the search engine.

\section{CE Ranking Considering Mutual Influences: Diversity Ranking}
\label{sec-divRanking}
An assumption in CE ranking is that the entities are mutually independent as we pointed out in Section~\ref{sec-clickModel}. In other words, the three parameters---$U(e)$, $C(e)$ and $\gamma(e)$---of an entity do not depend on other entities in the ranked list. In this section we relax this assumption and analyze the implications. Since the nature of the mutual influence may vary for different problems, we base our analysis on a specific well known  problem---ranking considering diversity~\cite{carterette-analysis,agrawal2009diversifying,rafiei-diversifying}.

Diversity ranking accounts for the fact that the utility of an entity is reduced by the presence of a similar  entity above in the ranked list. This is a typical example of the mutual influence between the entities. All the existing objective functions for the diversity ranking are known to be NP-Hard~\cite{carterette-analysis}. We  analyze a most basic form of diversity ranking to explain why this is a fundamentally hard problem.

We modify the objective function in Equation~\ref{objective-eq} slightly to distinguish between the stand-alone utilities and the residual utilities---utility of an entity in the context of other entities in the list---as,
\begin{equation}E(U)=\sum_{i=1}^{N}U_r(e_i)P_c(e_i)\label{div-objective-eq}\end{equation}
where $U_r(e_i)$ denotes the residual utility.

We consider a simple case of diversity ranking problem by considering a set of entities---all having the same utilities, perceived relevances and abandonment probabilities. Some of these entities may be repeating. If an entity in  the ranked list is same as the entity in the list above, residual utility of that entity becomes zero. In this case, it is intuitive that the optimal ranking is to place maximum number of  pair wise dissimilar entities in the top slots. The theorem below shows that even in this simple case the optimal ranking is NP-Hard.

\begin{thm}Diversity ranking  optimizing
expected utility in Equation \ref{div-objective-eq}
is NP-Hard.\label{thmnphard}\end{thm}

\emph{Proof Sketch:}
The proof is by reduction from the independent set problem. See Appendix \ref{AppendixNPH} for the complete proof.
$\Box$

Moreover, the proof by reduction from independent set problem has  more severe implications than NP-Hardness as shown in the following corollary,

\newtheorem{corrollary}{Corollary}\begin{corrollary}The constant approximation algorithm for ranking considering diversity is hard.\label{thmcorr}\end{corrollary}
{\em Proof:} The proof of NP-Hardness theorem above shows that the independent set problem is a special case of diversity ranking. This implies that a constant ratio approximation algorithm for the optimal diversity ranking would be a constant ratio approximation algorithm for the independent set problem. Since  constant ratio approximation of the independent set is known to be hard (\emph{cf.} Garey and Johnson~\cite{gareyIS} and H{\aa}stad \cite{hastad1996cha}) the corollary follows. To define hard, in his landmark paper H{\aa}stad proved that independent set cannot be solved within $n^{1-\epsilon}$ for $\epsilon > 0 $ unless all problems in $NP$ are solvable in probabilistic polynomial time, which is widely believed to be not possible.\footnote{This belief is almost as strong as the belief $P\ne NP$}
$\Box$

This section shows that the optimal ranking considering mutual influences of parameters is hard. We are leaving formulating approximation algorithms (not necessarily constant ratio) for future research.

Beyond proving the intractability of mutual influence ranking, we believe that intractability of the simple scenario here explains why all diversity rankings  are likely to be intractable. Further the proof based on the reduction from the well explored independent set problem may help in adapting approximations algorithms from graph theory.

\section{Conclusion and Future Work}
We approach the web ranking as a utility maximization based on user's click model, and derive the optimal ranking---namely CE ranking. The ranking is simple and intuitive; and optimal (for the given utilities) considering  perceived relevance and abandonment probability of user behavior.

 For specific assumptions on parameters, the ranking function reduces to a taxonomy of ranking functions in multiple ranking domains. The enumerated taxonomy will help to decide optimal ranking for a specific user behavior. In addition, the taxonomy shows that the existing document and ad ranking strategies are special cases of the proposed ranking function under specific assumptions.   

 To apply CE ranking to ad auctions, we incorporate a  second price based pricing. The resulting CE mechanism has a Nash Equilibrium which simultaneously optimizes search engine and advertiser revenues. CE mechanism is revenue dominant over VCG for the same bid vectors, and has an equilibrium which is revenue equivalent with the truthful equilibrium of VCG.

Finally, we relax the assumption of independence between entities in CE ranking and consider diversity ranking. The ensuing analysis revels that diversity ranking is an inherently hard problem; since even the basic formulations are NP-Hard with unlikely constant ratio approximation algorithms.

Our previous simulation studies suggest significant improvement in profits by CE ranking over existing ranking strategies~\cite{raju2008optimal}. As a future research, assessing profits by CE mechanism on a large scale search engine click log will quantify improvement in a real environment. Learning and prediction of abandonment probability from click logs as well as by  parametric learning are interesting problems. The suggested ranking is optimal for other web ranking scenarios with similar click models---like products and friends recommendations---and may be extended to these problems. Further, effective approximation schemes for diversity ranking based on similarity with the independent set problem may be investigated.

\bibliographystyle{acmtrans}
\bibliography{profitmax}
\renewcommand{\theequation}{A-\arabic{equation}}
\renewcommand{\thesubsection}{A-\arabic{subsection}}
\setcounter{equation}{0}  
{\centering
\section*{APPENDIX}}
\subsection{Proof of Theorem \ref{thmRF}}
\label{appendixRF}
\begin{proof}
Consider results $e_i$ and $e_{i+1}$ in positions $i$ and $i+1$ respectively. Let $\mu_i= \gamma (e_i)+C(e_i)$ for notational convenience.
The total expected utility from $e_i$ and $e_{i+1}$ when $e_i$ is placed above $e_{i+1}$ is
\begin{displaymath}
\prod_{j=1}^{i-1}{(1-\mu_j)}\left[ U(e_i) C(e_i) + (1- \mu_i)  U(e_{i+1}) C(e_{i+1})\right]
\end{displaymath}
If the order of $e_i$ and $e_{i+1}$ are inverted by  placing $e_i$ above $e_{i+1}$, the expected utility from these entities will be,
\begin{displaymath}
\prod_{j=1}^{i-1}(1-\mu_j)\left[ U(e_{i+1}) C(e_{i+1}) + (1- \mu_{i+1})  U(e_i) C(e_i))\right]
\end{displaymath}
Since utilities from all other results in the list will remain the same, the expected utility of placing $e_i$ above $e_{i+1}$ is greater than inverse placement \emph{iff}
\begin{eqnarray}
U(e_i) C(e_i) + (1- \mu_i)  U(e_{i+1}) C(e_{i+1}) &\ge&  U(e_{i+1}) C(e_{i+1}) + (1- \mu_{i+1})  U(e_i) C(e_i) \nonumber \\
& \Updownarrow & \nonumber \\
\frac{U(e_i) C(e_i)}{\mu_i}  &\ge&  \frac{U(e_{i+1}) C(e_{i+1})}{\mu_{i+1}} \nonumber
\end{eqnarray}
This means if entities are ranked in the descending order of $\frac{U(e) C(e)}{C(e)+\gamma(e)}$ any inversions will reduce the profit. Otherwise ranking by  $\frac{U(e) C(e)}{C(e)+\gamma(e)}$ is optimal.
\end{proof}

\subsection{Proof of Theorem \ref{thm-CEsameWi}}
\label{appendix-CESameWi}
\begin{proof}
Without  loss of generality, we assume that $a_i$ refers to ad in the position $i$ in the descending order of $w_ib_i$.
\begin{eqnarray}
CE_i& = & \frac{p_ic_i}{\mu_i} \nonumber \\
    & = & \frac{b_{i+1}c_{i+1} \mu_i }{\mu_{i+1} c_i} \frac{c_i}{\mu_i}  \nonumber \\
    & = &  \frac{b_{i+1}c_{i+1}}{\mu_{i+1}} \nonumber \\
    & = &  w_{i+1}b_{i+1} \nonumber \\
    & \ge &  w_{i+2}b_{i+2} = CE_{i+1}\nonumber
\end{eqnarray}
\end{proof}

\subsection{Proof of Theorem \ref{thmNash}}
\label{appedixNash}
Let there are $n$ advertisers. Without loss of generality, let us assume that advertisers are indexed in the descending order of $\frac{v_i c_i}{\mu_i}$. We prove equilibrium in two steps.

\noindent
\textbf{Step 1:} Prove that
\begin{equation}w_i b_i \ge w_{i+1}b_{i+1} \label{eq-weightedEqOrder} \end{equation}

\begin{proof}
\begin{displaymath}
w_i  b_i            =            \frac{b_i c_i}{\mu_i}
\end{displaymath}
Expanding $b_i$ by Equation~\ref{eq-nashEqBids},
\begin{eqnarray}
  w_i b_i    &       =           &  v_ic_i+(1-\mu_i)\frac{b_{i+1}c_{i+1}}{\mu_{i+1}}  \nonumber \\
             &       =           &  v_ic_i+(1-\mu_i)w_{i+1}b_{i+1}  \nonumber \\
             &       =           &  \frac{v_ic_i}{\mu_i} \mu_i+(1-\mu_i)w_{i+1} b_{i+1}  \nonumber
\end{eqnarray}
Notice that $w_i b_i$ is a convex linear combination of $w_{i+1}b_{i+1}$ and $\frac{v_ic_i}{\mu_i}$. This means that the value of $w_ib_i$ is in between (or equal to) the values of $w_{i+1}b_{i+1}$ and $\frac{v_ic_i}{\mu_i}$. Hence to prove that $w_ib_i\ge w_{i+1}b_{i+1}$ all we need to prove is that $\frac{v_ic_i}{\mu_i} \ge w_{i+1}b_{i+1}$. This inductive proof is given below.

\noindent
\textbf{Induction hypothesis: } Assume that
\begin{displaymath}\forall_{i \ge  j } \frac{v_ic_i}{\mu_i} \ge w_{i+1} b_{i+1}\end{displaymath}

\noindent
\textbf{Base case: } Prove for $i=N$  i.e. for the bottommost  ad.
\begin{displaymath}\frac{v_{N-1} c_{N-1}}{\mu_{N-1}}  \ge w_{N} b_{N}\end{displaymath}
Assuming $\forall_{i > N} b_i = 0 $
\begin{displaymath} w_Nb_N = v_N c_N \le \frac{v_N c_N}{\mu_N} \mbox{ (as $\mu_N \le 1$)}  \le \frac{v_{N-1}c_{N-1}}{\mu_{N-1}} \mbox{ (by the assumed order i.e. by $\frac{v_i c_i}{\mu_i}$)}\end{displaymath}

\noindent
\textbf{Induction: }Expanding $w_jb_j$ by Equation~\ref{eq-nashEqBids},
\begin{displaymath}w_j b_j =   \frac{v_j c_j}{\mu_j} \mu_j+(1-\mu_j)w_{j+1} b_{j+1}  \end{displaymath}
 $w_jb_j$ is the convex linear combination, i.e $  \frac{v_j c_j}{\mu_j} \ge w_{j}b_{j} \ge w_{j+1} b_{j+1}$, as we know that $\frac{v_jc_j}{\mu_j}  \ge w_{j+1}b_{j+1}$ by induction hypothesis. Consequently,
\begin{displaymath}w_j b_j \le \frac{v_j c_j}{\mu_j} \le \frac{v_{j-1}c_{j-1}}{\mu_{j-1}}\mbox{ (by the assumed order)}\end{displaymath}
 This completes the induction.
\end{proof}

Since advertisers are ordered by $w_ib_i$ for pricing, the above proof says that the pricing order is the same as the assumed order in this proof (i.e. ordering by $\frac{v_i  c_i}{\mu_i}$). Consequently,
\begin{displaymath}p_i = \frac{b_{i+1}c_{i+1} \mu_i}{\mu_{i+1}c_i} \end{displaymath}

As corollary of Theorem~\ref{thm-CEsameWi} we know that $CE_i \ge CE_{i+1}$.

In the second step we prove the envy free equilibrium using results in Step~1.
\\
\\
\noindent
\textbf{Step 2: No advertiser can increase his profit by changing his bids unilaterally}
\begin{proof}[ of Envy Freeness to Advertisers Below]
\label{proof-envrfreeBelow}
 In the first step let us prove that ad $a_i$ can not increase his profit by decreasing his bid to move to a position $j \ge i$ below.

\noindent
\textbf{Inductive hypothesis: } Assume  true for $ i \le j \le m $.

\noindent
\textbf{Base Case:} Trivially true for $j=i$.

\noindent
\textbf{Induction:} Prove that the expected profit of $a_i$ at $m+1$  is less or equal to the expected profit of $a_i$ at $i$.

Let $\rho_k$ denotes  the amount paid by $a_i$ when he is at the position $k$. By inductive hypothesis, the expected profit at $m$ is less or equal to the expected profit at $i$. So we just need to prove that the expected profit at $m+1$ is less or equal to the expected profit at $m$. i.e.
\begin{displaymath}
\frac{(v_i - \rho_{m})}{(1-\mu_i)}\prod_{l=1}^{m}(1-\mu_l)  \ge   \frac{(v_i - \rho_{m+1})}{(1-\mu_i)}\prod_{l=1}^{m+1}(1-\mu_l)
\end{displaymath}
Canceling the common terms,
\begin{equation}
v_i - \rho_{m}   \ge   (v_i - \rho_{m+1})( 1 - \mu_{m+1})
\label{eq-proofProfitBelow}
\end{equation}
$\rho_m$---the price charged to $a_i$ at position $m$---is based on the Equations~\ref{eq-nextPricePricing} and ~\ref{eq-nashEqBids}. Since the $a_i$ is moving downward, $a_i$ will occupy position $m$ by shifting  ad $a_m$ upwards. Hence the ad just below $a_i$ is $a_{m+1}$. Consequently, the price charged to $a_i$ when it is at the $m^{th}$ position is,
\begin{displaymath}
\rho_m  = \frac{b_{m+1}c_{m+1}\mu_{i}}{\mu_{m+1}c_{i}} = \frac{\mu_i}{c_i}\left[ v_{m+1}c_{m+1}+ (1-\mu_{m+1})\frac{b_{m+2}c_{m+2}}{\mu_{m+2}} \right]
\end{displaymath}
Substituting for $\rho_m$ and $\rho_{m+1}$ in Equation~\ref{eq-proofProfitBelow},
\begin{eqnarray}
 v_i \! - \! \frac{\mu_i}{c_i}\left[v_{m+1}c_{m+1}+ (1-\mu_{m+1})\frac{b_{m+2}c_{m+2}}{\mu_{m+2}}\right]\!\!&\ge&\!\! \left(v_i-\frac{\mu_i}{c_i}\left[v_{m+2}c_{m+2}+ \frac{}{} \right.\right. \nonumber \\
  && \left. \left. (1-\mu_{m+2})\frac{b_{m+3}c_{m+3}}{\mu_{m+3}}\right]\right)\!( 1 \! -\!\mu_{m+1}) \nonumber
 \end{eqnarray}
 Simplifying, and  multiplying  both sides by $-1$
 \begin{eqnarray}
 \frac{\mu_i}{c_i}\left[v_{m+1}c_{m+1}+(1-\mu_{m+1})\frac{b_{m+2}c_{m+2}}{\mu_{m+2}}\right]& \le &v_i\mu_{m+1}+  \frac{\mu_i}{c_i} ( 1 - \mu_{m+1}) \left[\frac{}{}v_{m+2}c_{m+2}+ \right.  \nonumber \\
 && \;\;\;\;\;\;\left. (1-\mu_{m+2})\frac{b_{m+3}c_{m+3}}{\mu_{m+3}}\right]  \nonumber
 \end{eqnarray}
 Substituting by $b_{m+2}$ from Equation~\ref{eq-nashEqBids} on RHS.
 \begin{displaymath}
 \frac{\mu_i}{c_i}\left[v_{m+1}c_{m+1}+(1-\mu_{m+1})\frac{b_{m+2}c_{m+2}}{\mu_{m+2}}\right] \le v_i\mu_{m+1}+  \frac{\mu_i}{c_i} ( 1 - \mu_{m+1}) \frac{b_{m+2} c_{m+2}}{\mu_{m+2}}\nonumber \\
 \end{displaymath}
  Canceling out the common terms on both sides,
 \begin{eqnarray}
 \frac{\mu_i}{c_i} v_{m+1}c_{m+1}& \le &v_i\mu_{m+1} \nonumber \\
 &\Updownarrow& \nonumber \\
 \frac{v_{m+1}c_{m+1}}{\mu_{m+1}} & \le & \frac{v_ic_i}{\mu_i} \nonumber
\end{eqnarray}
Which is true by the assumed order as $m \ge i$
\end{proof}

\noindent
Inductive proof for $ m \le i$ is somewhat similar and enumerated below.

\noindent
\textbf{Inductive hypothesis: } Assume  true for $j \le m $.

\noindent
\textbf{Base Case:} Trivially true for $j=i$.

\begin{proof}[of Envy freeness to the ad one above ]

The case in which $a_i$ increase his bid to move one position up i.e. to $i-1$ is a special case and need to be proved separately. In this case, by moving a single slot up, the index of the ad below $a_i$ will change from $i+1$ to $i-1$ (a difference of two). For all other movements of $a_i$ to a position one above or one below, the index of the advertisers below will change only by one. Since the amount paid by $a_i$ depends on the ad below $a_i$, this case warrants a slightly different proof,
\begin{eqnarray}
(v_i - \rho_{i})\prod_{l=1}^{i-1}(1-\mu_l) & \ge &  (v_i - \rho_{m-1})\prod_{l=1}^{i-2}(1-\mu_l) \nonumber \\
&\Updownarrow& \nonumber\\
(v_i - \rho_{i})(1-\mu_{i-1})  & \ge &  v_i - \rho_{i-1} \nonumber
\end{eqnarray}
Expanding $\rho_i$ is straight forward.To expand $\rho_{i-1}$, note that when $a_i$ has moved upwards to $i-1$, the ad just below $a_i$ is  $a_{i-1}$. Since $a_{i-1}$ has not changed its bids, the  $\rho_{i-1}$ can be expanded  as $\frac{\mu_i}{c_i}\left[ v_{i-1}c_{i-1}+ (1-\mu_{i-1})\frac{b_{i}c_{i}}{\mu_{i}}\right]$. Substituting for $\rho_i$ and $\rho_{i-1}$,
\begin{eqnarray}
 \left( v_i - \frac{\mu_i}{c_i}\left[v_{i+1}c_{i+1}+ \frac{}{} \right.\right. \;\;\;\;\;\; & \ge & v_i-\frac{\mu_i}{c_i}\left[ v_{i-1}c_{i-1}+ \frac{}{} \right. \nonumber  \\
\left.\left. (1-\mu_{i+1})\frac{b_{i+2}c_{i+2}}{\mu_{i+2}}\right] \right) (1-\mu_{i-1}) &  & \;\;\;\;\;\left.(1-\mu_{i-1})\frac{b_{i}c_{i}}{\mu_{i}}\right] \nonumber
\end{eqnarray}
 Simplifying and multiplying by $-1$
 \begin{eqnarray}
 v_i\mu_{i-1}+\frac{\mu_i}{c_i}\left[v_{i+1}c_{i+1} +  \frac{}{} \right. \;\;\;\;\;\;& \le &  \frac{\mu_i}{c_i} \left[v_{i-1}c_{i-1}+ (1-\mu_{i-1})\frac{b_{i}c_{i}}{\mu_{i}}\right]  \nonumber \\
 \left. (1-\mu_{i+1})\frac{b_{i+2}c_{i+2}}{\mu_{i+2}}\right]  (1-\mu_{i-1}) && \nonumber
 \end{eqnarray}
 Substituting $b_{i+1}$ from Equation~\ref{eq-nashEqBids}
  \begin{eqnarray}
  v_i\mu_{i-1}+\frac{\mu_i}{c_i} \frac{b_{i+1}c_{i+1}}{\mu_{i+1}} (1-\mu_{i-1})& \le &  \frac{\mu_i}{c_i} \left[v_{i-1}c_{i-1}+ (1-\mu_{i-1})\frac{b_{i}c_{i}}{\mu_{i}}\right]  \nonumber \\
  &\Updownarrow& \nonumber \\
 v_i\mu_{i-1}+\frac{\mu_i}{c_i} (1-\mu_{i-1}) \frac{b_{i+1}c_{i+1}}{\mu_{i+1}} & \le &     \frac{\mu_i v_{i-1}c_{i-1} }{c_i} + \frac{\mu_i}{c_i}(1-\mu_{i-1})\frac{b_{i}c_{i}}{\mu_{i}} \nonumber
 \end{eqnarray}
 We now prove that both the terms in RHS are greater or equal to the corresponding terms in LHS separately.
\begin{eqnarray}
v_i\mu_{i-1}  & \le   &   \frac{\mu_i v_{i-1}c_{i-1} }{c_i}  \nonumber \\
 & \Updownarrow   &   \nonumber \\
 \frac{v_ic_i}{\mu_i}  & \le   &   \frac{ v_{i-1}c_{i-1} }{\mu_{i-1}}  \nonumber
\end{eqnarray}
Which is true by our assumed order.

\noindent
Similarly,
\begin{eqnarray}
\frac{\mu_i}{c_i} (1-\mu_{i-1}) \frac{b_{i+1}c_{i+1}}{\mu_{i+1}} & \le & \frac{\mu_i}{c_i}(1-\mu_{i-1})\frac{b_{i}c_{i}}{\mu_{i}} \nonumber  \\
& \Updownarrow   &   \nonumber \\
 \frac{b_{i+1}c_{i+1}}{\mu_{i+1}} & \le & \frac{b_{i}c_{i}}{\mu_{i}} \nonumber
\end{eqnarray}
Which is true by Equation~\ref{eq-weightedEqOrder} above. This completes the proof for this  case.
\end{proof}

\noindent
\textbf{Induction:}
Prove that the expected profit at $m-1$ is less or equal to the expected profit at $m$. The proof is similar to the induction for the case $m>i$.
\begin{proof}
Base case is trivially true.
\begin{displaymath}
(v_i - \rho_{m})\prod_{l=1}^{m-1}(1-\mu_l)  \ge   (v_i - \rho_{m-1})\prod_{l=1}^{m-2}(1-\mu_l)
\end{displaymath}
Canceling common terms,
\begin{displaymath}
(v_i - \rho_{m})(1-\mu_{m-1})   \ge   v_i - \rho_{m-1}
\end{displaymath}
In this case, note that $a_i$ is moving upwards. This means that $a_i$ will occupy position $m$ by pushing the ad originally at $m$ one position downwards. Hence the original ad at $m$ is the one just below $a_i$ now. i.e.
\begin{displaymath}
\rho_m  = \frac{b_{m}c_{m}\mu_{i}}{\mu_{m}c_{i}} = \frac{\mu_i}{c_i}\left[ v_{m}c_{m}+ (1-\mu_{m})\frac{b_{m+1}c_{m+1}}{\mu_{m+1}} \right]
\end{displaymath}
Substituting for $\rho_m$ and $\rho_{m-1}$
\begin{eqnarray}
 \left( v_i - \frac{\mu_i}{c_i}\left[v_{m}c_{m}+ \frac{}{} \right.\right. \;\;\;\;\;\; & \ge & v_i-\frac{\mu_i}{c_i}\left[ v_{m-1}c_{m-1}+ \frac{}{} \right. \nonumber  \\
\left.\left. (1-\mu_{m})\frac{b_{m+1}c_{m+1}}{\mu_{m+1}}\right] \right) (1-\mu_{m-1}) &  & \;\;\;\;\;\left.(1-\mu_{m-1})\frac{b_{m}c_{m}}{\mu_{m}}\right] \nonumber
\end{eqnarray}
 Simplifying and multiplying by $-1$
 \begin{eqnarray}
 v_i\mu_{m-1}+\frac{\mu_i}{c_i}\left[v_{m}c_{m} +  \frac{}{} \right. \;\;\;\;\;\;& \le &  \frac{\mu_i}{c_i} \left[v_{m-1}c_{m-1}+ (1-\mu_{m-1})\frac{b_{m}c_{m}}{\mu_{m}}\right]  \nonumber \\
 \left. (1-\mu_{m})\frac{b_{m+1}c_{m+1}}{\mu_{m+1}}\right]  (1-\mu_{m-1}) && \nonumber
 \end{eqnarray}

 Substituting by $b_{m}$ from Equation~\ref{eq-nashEqBids}
 \begin{displaymath}
 v_i\mu_{m-1}+\frac{\mu_i}{c_i}\frac{b_mc_m}{\mu_m} (1-\mu_{m-1}) \le \frac{\mu_i}{c_i} \left[v_{m-1}c_{m-1}+ (1-\mu_{m-1})\frac{b_{m}c_{m}}{\mu_{m}}\right]
 \end{displaymath}
 Canceling common terms,
 \begin{eqnarray}
v_i\mu_{m-1} & \le &\frac{\mu_i}{c_i} v_{m-1}c_{m-1}  \nonumber \\
&\Updownarrow& \nonumber \\
\frac{v_ic_i}{\mu_i} & \le & \frac{v_{m-1}c_{m-1}}{\mu_{m-1}}  \nonumber
\end{eqnarray}
Which is true by the assumed order as $ m < i$.
\end{proof}

\subsection{Proof of Theorem \ref{thm-CERevnueDominanceVCG}}
\label{appendix-CERevnueDominanceVCG}
\begin{proof}
VCG payment of the ad at position $i$ (i.e. $a_i$) is equal to the reduction in utility of the ads below due to the presence of  $a_i$. For each user viewing the list of ads (i.e. for unit view probability), the total expected loss of ads below $a_i$ due to $a_i$ is,
\begin{eqnarray}
p_{i}^{V_u} & = & \frac{1}{1-\mu_i}\sum_{j=i+1}^{n} b_j c_j\prod_{k=1}^{j-1}(1-\mu_k) - \sum_{j=i+1}^{n} b_j c_j\prod_{k=1}^{j-1}(1-\mu_k)  \nonumber \\
 & = & \frac{\mu_i}{1-\mu_i}\sum_{j=i+1}^{n} b_j c_j\prod_{k=1}^{j-1}(1-\mu_k) \nonumber \\
 & = & \frac{\mu_i}{1-\mu_i} \prod_{k=1}^{i}(1-\mu_k)\sum_{j=i+1}^{n} b_jc_j\prod_{k=i+1}^{j-1}(1-\mu_k) \nonumber \\
 & = & \mu_i \prod_{k=1}^{i-1}(1-\mu_k)\sum_{j=i+1}^{n} b_jc_j\prod_{k=i+1}^{j-1}(1-\mu_k) \nonumber
\end{eqnarray}
This is the expected lose per user browsing the ad list. Pay per click should be equal to the lose per click. To calculate the pay per click, we divide by the  click probability of $a_i$. i.e.
\begin{eqnarray}
p_{i}^{V} & = & \frac{ \mu_i \prod_{k=1}^{i-1}(1-\mu_k)\sum_{j=i+1}^{n} b_jc_j\prod_{k=i+1}^{j-1}(1-\mu_k)}{ c_i \prod_{k=1}^{i-1}(1-\mu_k)}  \nonumber \\
     & = &  \frac{\mu_i}{c_i} \sum_{j=i+1}^{n} b_jc_j\prod_{k=i+1}^{j-1}(1-\mu_k) \nonumber
\end{eqnarray}

Converting to recursive form,
\begin{eqnarray}
p_{i}^{V} & = &  \frac{ b_{i+1} \mu_i}{c_i}  c_{i+1}  +  (1-\mu_{i+1}) \frac{\mu_i c_{i+1}}{c_i\mu_{i+1}}p_{i+1}^{V}   \nonumber   \\
& = &   \frac{b_{i+1} \mu_i  c_{i+1}}{c_i \mu_{i+1}}  \mu_{i+1}  +  (1-\mu_{i+1}) \frac{\mu_i c_{i+1}}{c_i\mu_{i+1}}p_{i+1}^{V}  \label{eq-VCGrecursivePay}
\end{eqnarray}
For the  $CE$ mechanism payment from Equation~\ref{eq-nextPricePricing} is,
\begin{displaymath}
p_{i}^{CE}  =   \frac{b_{i+1}c_{i+1}\mu_i}{\mu_{i+1}c_i}
\end{displaymath}

 Note that $p_i^V$ is convex combination of $P_i^{CE}$ and  $\frac{\mu_i c_{i+1}}{c_i\mu_{i+1}}p_{i+1}^{V}$, and hence is between these two values. To prove that $p_i^{CE} \ge p_i^V  $ all we need to prove is that $ P_i^{CE} \ge \frac{\mu_i c_{i+1}}{c_i\mu_{i+1}}p_{i+1}^{V} \Leftrightarrow b_{i} \ge p_i^{V}$. This directly follows from individual rationality property of VCG. Alternatively,  a simple recursion with base case as  $p_N^V=0$ (bottommost ad) will prove the same. Note that we consider only the ranking (not selection), and hence the VCG pricing of the bottommost ad in the ranking is zero.
\end{proof}

\subsection{Proof of Theorem \ref{thm-VCGrevnueEqiv}}
\label{appendix-VCGRevnueEquiv}
\begin{proof}
Rearranging Equation~\ref{eq-VCGrecursivePay} and substituting true values for bid amounts,
\begin{eqnarray}
p_{i}^{V}& = &  \frac{\mu_i}{c_i} \left[ v_{i+1} c_{i+1}  + \frac{(1-\mu_{i+1})c_{i+1}}{\mu_{i+1}}p_{i+1}^{V} \right]  \nonumber
\end{eqnarray}
For the $CE$ mechanism, substituting  equilibrium bids from  Equation~\ref{eq-nashEqBids} in payment  (Equation~\ref{eq-nextPricePricing}),
\begin{displaymath}
p_{i}^{CE}  =   \frac{b_{i+1}c_{i+1}\mu_i}{\mu_{i+1}c_i} =  \frac{\mu_i}{c_i} \left[ v_{i+1} c_{i+1} +(1-\mu_{i+1}) \frac{b_{i+2}c_{i+2}}{\mu_{i+2}} \right] \end{displaymath}
 Rewriting  $b_{i+2}$ in terms of $p_{i+1}$,
\begin{eqnarray}
p_{i}^{CE} & = &  \frac{\mu_i}{c_i} \left[ v_{i+1} c_{i+1} +\frac{(1-\mu_{i+1})c_{i+1}}{\mu_{i+1}}  p_{i+1}^{CE} \right]   \nonumber \\
      & = &  p_{i}^{V} \;\;\;\;\mbox{(\emph{iff} $p_{i+1}^{V}=p_{i+1}^{CE}$)} \nonumber
\end{eqnarray}
 Ad at the bottommost position pays same amount zero, a simple recursion will prove that the payment for all positions for both VCG and the proposed equilibrium is the same.
\end{proof}

\subsection{Proof of Theorem \ref{thmnphard}}
\label{AppendixNPH}
\begin{proof} Independent set problem can be formulated as a ranking problem considering
similarities. Consider an unweighed graph G of $n$ vertices $ \{e_1,e_2,..e_n\} $ represented as an
adjacency matrix. This conversion is clearly polynomial time. Now, consider the values in the
adjacency matrix as  the similarity values between the entities to be
ranked. Let the entities have the same utilities, perceive relevances  and abandonment probabilities. In this set of $n$ entities from $\{e_1,
e_2,..,e_n\}$, clearly the optimal ranking will have $k$ pairwise independent entities as the top $k$ entities for a maximum possible value of $k$. But the set of $k$ independent entities corresponds
to the maximum independent set in graph G.
\end{proof}

\end{document}